\documentclass[11pt]{article}

\usepackage[utf8]{inputenc}
\usepackage[T1]{fontenc}
\usepackage{graphicx}
\usepackage{algorithm}
\usepackage{algpseudocode}
\usepackage{amssymb}
\usepackage{amsmath}
\usepackage{hyperref}
\usepackage[margin=1in]{geometry}
\usepackage{booktabs}
\usepackage{siunitx}
\usepackage{pgfplots}
\usepackage{pgfplotstable}
\pgfplotsset{compat=1.18}
\usepackage{xcolor}
\usepackage{caption}
\usepackage{subcaption}
\usepackage{multirow}
\usepackage{placeins}
\usepackage{float}

\usepackage{adjustbox}
\usepackage{lmodern}
\usepackage{authblk}          
\usepackage{amsthm}           
\usetikzlibrary{plotmarks}
\newtheorem{theorem}{Theorem}[section]
\newtheorem{lemma}[theorem]{Lemma}
\newtheorem{corollary}[theorem]{Corollary}

\theoremstyle{definition}

\theoremstyle{remark}

\newcommand{\keywords}[1]{%
  \smallskip\par\noindent\textbf{Keywords:}\ #1\par}

\definecolor{col1}{RGB}{31,119,180}
\definecolor{col2}{RGB}{255,127,14}
\definecolor{col3}{RGB}{44,160,44}
\definecolor{col4}{RGB}{214,39,40}
\definecolor{col5}{RGB}{148,103,189}
\definecolor{col6}{RGB}{140,86,75}
\definecolor{col7}{RGB}{23,190,207}
\definecolor{col8}{RGB}{188,189,34}
\definecolor{col9}{RGB}{227,119,194}

\title{Random Proposals: A Softmax-Based Local-Improvement Framework
       for Maximum Weighted Matching}

\author[1]{Ahmed M. Alzuhair\thanks{\texttt{alzuhairahmed@gmail.com}}}
\author[1]{Ahmed Alherz\thanks{\texttt{alherz@kfupm.edu.sa}}}

\affil[1]{Department of Information and Computer Science, King Fahd University of Petroleum and Minerals (KFUPM), Dhahran, Saudi Arabia}

\date{}

\begin{document}

\maketitle

\begin{abstract}
We propose a randomized local-improvement algorithm for the Maximum Weighted Matching (MWM) problem.
Our method introduces a softmax-based biased sampling mechanism that achieves local $\varepsilon$-dominance
and yields an expected $\frac{1}{2}-\varepsilon$ approximation ratio.
We prove convergence guarantees and show that the algorithm runs in
$O\!\left(m\log(1/\varepsilon)/p_{\min}\right)$ time,
where $p_{\min}$ is the minimum softmax proposal probability over all edges;
under mild conditions on the bias parameter and weight range, this simplifies to $O(m\log(1/\varepsilon))$.
The framework provides a tunable tradeoff between convergence speed and approximation quality.
\end{abstract}

\section{Introduction}

Maximum Weighted Matching (MWM) is a fundamental problem in combinatorial optimization, with broad applications in resource allocation, network design, and assignment problems. Given a weighted graph $G=(V,E,w)$, the goal is to find a set of pairwise non-adjacent edges of maximum total weight. Beyond its practical relevance, MWM has driven significant progress in the study of graph algorithms since Edmonds~\cite{edmonds1965paths} established the polynomial-time solvability of the matching problem in general graphs in 1965. Since then, a long chain of theoretical and algorithmic advances has reflected the enduring importance of MWM in both theory and practice.

There is, however, a noticeable gap between the theoretical guarantees and the practical deployment of MWM algorithms. Although exact algorithms are mathematically rigorous, they often rely on complex data structures that are difficult to implement and prohibitively expensive on very large graphs. This has motivated research on approximation algorithms, which trade optimality for scalability. Several algorithms achieve a $1/2$ approximation ratio, including greedy-style, path-growing, and proposal-based methods. Despite sacrificing optimality, these algorithms offer substantial reductions in runtime and are well suited for large-scale settings.

Other approaches incorporate randomization to improve efficiency, as in~\cite{pettie2004simple}. Meanwhile, augmentation-based approximations surpassing the $1/2$ barrier are known, reaching $2/3$ or even $1-\varepsilon$ in near-linear time. However, these improvements typically require intricate combinatorial analysis or scaling frameworks. It therefore remains natural to ask whether one can design an algorithm that preserves the simplicity of local-improvement methods while providing tunable approximation quality and near-linear expected runtime.

In this paper we propose Random Proposals (RP), a randomized proposal-based algorithm for MWM. The algorithm uses a softmax sampling rule~\cite{bishop2006prml}: in each round, every vertex independently proposes to one of its neighbors with probability proportional to the softmax, with temperature $\beta$, of the incident edge weights. This mechanism biases the process toward heavier edges while retaining enough randomness to explore local configurations and maintain a favorable runtime. A user-specified round count $K$ directly controls the quality-runtime tradeoff.

The main contributions of this paper are:

\begin{itemize}
    \item We introduce a tunable softmax-based randomized proposal algorithm for MWM.
    \item We analyze the algorithm under a specified number of rounds $K$ and bias parameter $\beta$, and show that with appropriate settings it achieves an expected $(1/2-\varepsilon)$ approximation ratio.
    \item We derive an explicit relationship between $K$ and the approximation error, making the quality-runtime tradeoff precise.
    \item We empirically evaluate the proposed algorithm against established $1/2$-approximation baselines, including Greedy and Suitor, on a diverse set of large benchmark graphs.
\end{itemize}

The rest of the paper is organized as follows. Section~2 reviews related work. Section~3 introduces preliminaries and notation. Section~4 presents the algorithm and its theoretical analysis. Section~5 reports the experimental evaluation. Section~6 concludes and discusses future directions.

\section{Related Work}

The classical exact solution to MWM in general graphs is due to Edmonds~\cite{edmonds1965paths}, whose blossom algorithm established polynomial-time solvability. Modern implementations achieve $O(n^3)$, with later refinements by Gabow~\cite{gabow1990data} reaching $O(m\sqrt{n}\log n)$ and by Duan and Pettie~\cite{duan2014scaling} achieving $O(m\sqrt{n}\log(nN))$. These implementations typically require complex data structures, making them less practical for very large graphs.

For bipartite graphs, the Hungarian method~\cite{kuhn1955hungarian} solves the assignment problem in $O(n^3)$. Successive shortest augmenting paths~\cite{edmonds1972theoretical} achieve $O(nm + n^2\log n)$, and the cost-scaling variant of Goldberg and Kennedy~\cite{goldberg1990finding} achieves $O(nm\log(nC))$, where $C$ is the maximum edge cost.

The expense of exact algorithms has motivated substantial research on approximation algorithms. A simple greedy algorithm~\cite{pothen2014approximation} that repeatedly selects the heaviest remaining edge achieves a $1/2$ approximation in $O(m\log n)$. The Path Growing algorithm~\cite{drake2003linear} attains $O(m+n)$ while maintaining the same ratio. The first linear-time $1/2$ approximation was given by Preis~\cite{preis1999linear}, though it can exhibit poor memory-access patterns in practice.

The Suitor algorithm of Manne and Halappanavar~\cite{manne2014suitor} is currently among the most practical $1/2$-approximation algorithms. It generalizes the stable-marriage idea~\cite{gale1962college} to weighted graphs: each vertex proposes to its heaviest available neighbor, and a proposal is accepted only if it improves the current match. The resulting chain of proposals stabilizes at a greedy matching and is highly parallelizable. In essence, all deterministic $1/2$-approximation algorithms match locally dominant edges through different structural arguments.

Beyond $1/2$, augmenting-path techniques push the ratio toward $2/3$. Pettie and Sanders~\cite{pettie2004simple} achieve $2/3-\varepsilon$ in $O(km)$ deterministically by applying bounded-length augmentations to a greedy matching. Duan and Pettie~\cite{duan2014linear} obtain a $1-\varepsilon$ approximation in $O(m\varepsilon^{-1}\log(1/\varepsilon))$, the strongest known near-linear result, by combining weight scaling with bounded augmenting structures.

Randomization has been incorporated into MWM approximation algorithms to simplify analysis and enable cleaner structural guarantees. The randomized algorithm of Pettie and Sanders~\cite{pettie2004simple} achieves $2/3-\varepsilon$ in expected $O(m\log(1/\varepsilon))$ by combining greedy initialization with randomly ordered bounded-length augmentations. Duan and Pettie further improved this to $3/4-\varepsilon$ in $O(m\log n\log(1/\varepsilon))$~\cite{duan2010approximating}.

In contrast to these augmentation- or scaling-based approaches, we propose a probabilistic local-dominance framework that provides explicit control over the approximation error through tunable parameters while maintaining near-linear expected runtime.

\section{Preliminaries}

Let $G=(V,E,w)$ be an undirected weighted graph, where $V$ is the vertex set, $E \subseteq V \times V$ is the edge set, and $w:E \to \mathbb{R}_{\geq 0}$ assigns non-negative weights. For a vertex $u \in V$, let $N(u)$ denote its neighborhood and $\deg(u) = |N(u)|$ its degree. We denote by $\Delta = \max_{v \in V}\deg(v)$ the maximum degree of the graph.

\subsection{Matching}

A matching $M \subseteq E$ is a set of edges such that no two edges in $M$ share a common vertex. For a vertex $u \in V$, we denote by $m(u)$ the vertex matched to $u$ in $M$; if no such vertex exists, we set $m(u) = \text{NULL}$ and $w(u,m(u)) = 0$. The weight of the matching is $w(M) = \sum_{e \in M} w(e)$. Let $M^*$ denote a maximum-weight matching.

\subsection{Softmax Sampling}

We define a softmax sampling distribution~\cite{bishop2006prml} over the neighbors of a vertex $u \in V$ as:
$$P(u \to v) = \frac{e^{\beta w(u,v)}}{\displaystyle\sum_{x\in N(u)}e^{\beta w(u,x)}},$$
where $\beta > 0$ is the bias parameter controlling the concentration toward high-weight edges. As $\beta \to \infty$ the distribution concentrates on the heaviest neighbor; as $\beta \to 0$ it approaches uniform sampling.

\begin{algorithm}[t]
\caption{Random Proposals}
\label{algorithm1}
\begin{algorithmic}[1]

\Require Graph $G = (V,E,w)$, parameter $\beta > 0$, number of rounds $K$
\Ensure Matching $M$

\State Initialize $M \gets \emptyset$
\For{$r = 1$ to $K$}
    \ForAll{$u \in V$}
        \State Sample $v \sim P(u \to \cdot)$
        \If{$w(u,v) > w(u,m(u)) + w(v,m(v))$}
            \State Remove $(u,m(u))$ and $(v,m(v))$ from $M$ (if they exist)
            \State Add $(u,v)$ to $M$
        \EndIf
    \EndFor
\EndFor

\State \Return $M$

\end{algorithmic}
\end{algorithm}

\section{Random Proposals Algorithm}

Algorithm~\ref{algorithm1} runs for $K$ rounds. In each round, every vertex independently samples a neighbor and applies an improving update if one exists. Without loss of generality, all weights are normalized so that $w(e) \geq 1$ for all $e \in E$. The algorithm terminates after exactly $K$ rounds and therefore always runs in finite time.

The acceptance condition $w(u,v) > w(u,m(u)) + w(v,m(v))$ is a strict improvement criterion: edge $(u,v)$ is accepted only if its weight exceeds the combined weight of the two edges it would displace. This ensures that the total matching weight increases with every accepted proposal and hence that proposals cannot cycle.

\subsection{Approximation Analysis}

\begin{lemma}
Let $u \in V$, $v^* = \arg\max_{x \in N(u)} w(u,x)$, and $X_u \sim P(u \to \cdot)$ be the neighbor proposed by $u$ under softmax sampling with bias $\beta$. Define $\mathbf{SoftmaxError}(u) = w(u,v^*) - w(u,X_u)$. Then:
$$
\mathbb{E}[\mathbf{SoftmaxError}(u)] \;\leq\; \frac{\ln\Delta}{\beta}.
$$
\end{lemma}

\begin{proof}
The expected edge weight of the softmax proposal is:
$$
\mathbb{E}[w(u,X_u)] = \frac{\displaystyle\sum_{x \in N(u)} w(u,x)\,e^{\beta w(u,x)}}{\displaystyle\sum_{x \in N(u)} e^{\beta w(u,x)}} = \frac{\partial}{\partial\beta} \ln \sum_{x \in N(u)} e^{\beta w(u,x)}.
$$
Since $\ln\sum_{x} e^{\beta w(u,x)}$ is convex in $\beta$, its derivative is bounded below by the slope of the function from $0$ to $\beta$:
$$
\frac{\partial}{\partial\beta} \ln \sum_{x} e^{\beta w(u,x)} \;\geq\; \frac{\ln\sum_{x} e^{\beta w(u,x)} - \ln\deg(u)}{\beta}.
$$
Applying the log-sum-exp lower bound $\ln\sum_{x} e^{\beta w(u,x)} \geq \beta w(u,v^*)$ gives:
$$
\mathbb{E}[w(u,X_u)] \;\geq\; w(u,v^*) - \frac{\ln\deg(u)}{\beta} \;\geq\; w(u,v^*) - \frac{\ln\Delta}{\beta}.
$$
Rearranging gives $\mathbb{E}[\mathbf{SoftmaxError}(u)] \leq \ln\Delta/\beta$, and the bound vanishes as $\beta \to \infty$.
\end{proof}

\begin{lemma}
Let $(u,v) \in E$ be an improving edge, i.e.\ $w(u,v) > w(u,m(u)) + w(v,m(v))$, that remains undiscovered after $K$ rounds. Let $l_{uv} = w(u,v) - w(u,m(u)) - w(v,m(v)) > 0$ be its improvement gap, and $p_{uv} = P(u \to v)$ the per-round proposal probability of $u$ toward $v$. Then:
$$
\mathbb{E}[\mathbf{StoppingError}(u,v)] \;\leq\; l_{uv}\,e^{-Kp_{uv}}.
$$
\end{lemma}

\begin{proof}
In each round, every vertex is visited and $u$ proposes to $v$ with probability $P(u \to v) = p_{uv}$. Since the weights are fixed across the execution of RP then $p_{uv}$ remains fixed across the $K$ rounds and is independent. Let $q \le K$ number of rounds such that $(u,v)$ is an improving edge then the probability of failing to discover $(u,v)$ is at most $(1-p_{uv})^K$. Using $1-x \leq e^{-x}$:

\[
\mathbb{E}[\mathrm{loss}_{uv}] \;\le\; l_{uv}\,(1-p_{uv})^K \;\leq\; l_{uv}\,e^{-Kp_{uv}}.
\]
\end{proof}

\begin{lemma}[Local $\varepsilon$-Dominance]
    Let $u \in V$, $v^* = \arg\max_{x \in N(u)} w(u,x)$, and $p_{\min} = \min_{(a,b)\in E} P(a\to b)$.
    With parameter choices $K = \Theta\!\left(\ln(1/\varepsilon)/p_{\min}\right)$
    and $\beta = O(\ln\Delta/\varepsilon)$, the following hold:
    \begin{enumerate}
        \item[(a)] \textbf{Unblocked case.}
              If $(u,v^*)$ is improving, i.e.\
              $w(u,v^*) > w(u,m(u)) + w(v^*,m(v^*))$, then
              \[
                  \mathbb{E}[w(u,m(u))] \;\geq\; (1-\varepsilon)\,w(u,v^*).
              \]
        \item[(b)] \textbf{Blocked case.}
              If $(u,v^*)$ is not improving, i.e.\
              $w(u,m(u)) + w(v^*,m(v^*)) \geq w(u,v^*)$, then
              \[
                  w(u,m(u)) + w(v^*,m(v^*)) \;\geq\; w(u,v^*).
              \]
    \end{enumerate}
\end{lemma}

\begin{proof}
    \textbf{Part (a).}
    Since $(u,v^*)$ is improving, it constitutes a valid acceptance target for $u$.
    The gap between $w(u,v^*)$ and $\mathbb{E}[w(u,m(u))]$ decomposes as:
    \[
        \mathbf{Error}(u) \;\leq\; \mathbf{StoppingError}(u,v^*) + \mathbf{SoftmaxError}(u).
    \]
    For the stopping error: since $w(u,m(u)), w(v^*,m(v^*)) \geq 0$, we have
    $l_{uv^*} = w(u,v^*) - w(u,m(u)) - w(v^*,m(v^*)) \leq w(u,v^*)$.
    By Lemma~2 with $p_{uv^*} \geq p_{\min}$ and $K = \Theta(\ln(1/\varepsilon')/p_{\min})$:
    \[
        \mathbf{StoppingError}(u,v^*)
        \;\leq\; l_{uv^*}\,e^{-Kp_{\min}}
        \;\leq\; w(u,v^*)\,e^{-Kp_{\min}}
        \;\leq\; \varepsilon'\,w(u,v^*).
    \]
    For the softmax error, by Lemma~1 with $\beta = O(\ln\Delta/\varepsilon')$:
    \[
        \mathbb{E}[\mathbf{SoftmaxError}(u)] \;\leq\; \frac{\ln\Delta}{\beta} \;\leq\; \varepsilon'.
    \]
    Since $w(e) \geq 1$ for all $e \in E$, we have $\varepsilon' \leq \varepsilon'\,w(u,v^*)$,
    so $\mathbf{Error}(u) \leq 2\varepsilon'\,w(u,v^*)$.
    Setting $\varepsilon' := \varepsilon/2$ yields $\mathbf{Error}(u) \leq \varepsilon\,w(u,v^*)$, and thus:
    \[
        \mathbb{E}[w(u,m(u))]
        \;\geq\; w(u,v^*) - \mathbf{Error}(u)
        \;\geq\; (1-\varepsilon)\,w(u,v^*).
    \]
    \textbf{Part (b).}
    The conclusion holds directly by hypothesis.
\end{proof}

\begin{theorem}[Approximation Guarantee]
Let $M$ be the matching returned by Algorithm~\ref{algorithm1} after $K$ rounds,
with parameters $K = \Theta\!\left(\ln(1/\varepsilon)/p_{\min}\right)$
and $\beta = O(\ln\Delta/\varepsilon)$,
where $p_{\min} = \min_{(u,v)\in E} P(u\to v)$.
Then:
\[
    \mathbb{E}[w(M)] \;\geq\; \frac{1-\varepsilon}{2}\,w(M^*).
\]
\end{theorem}

\begin{proof}
We show that for every edge $(u,v) \in M^*$,
\begin{equation}
    \label{eq:per-edge}
    \mathbb{E}[w(u,m(u)) + w(v,m(v))] \;\geq\; (1-\varepsilon)\,w(u,v).
\end{equation}
Summing \eqref{eq:per-edge} over $M^*$ and noting that each edge of $M$ is counted at most twice (once per endpoint):
\[
    \sum_{(u,v)\in M^*}\mathbb{E}[w(u,m(u))+w(v,m(v))]
    \;\geq\; (1-\varepsilon)\,w(M^*)
    \;\implies\;
    2\,\mathbb{E}[w(M)] \;\geq\; (1-\varepsilon)\,w(M^*),
\]
giving $\mathbb{E}[w(M)] \geq \tfrac{1-\varepsilon}{2}\,w(M^*)$.

It remains to prove \eqref{eq:per-edge}. Fix any $(u,v) \in M^*$.

\medskip
\noindent\textbf{Case 1} [\emph{$(u,v)$ is not improving}]:
$w(u,m(u)) + w(v,m(v)) \geq w(u,v)$.

The bound holds immediately and deterministically:
\[
    w(u,m(u)) + w(v,m(v)) \;\geq\; w(u,v) \;\geq\; (1-\varepsilon)\,w(u,v).
\]

\medskip
\noindent\textbf{Case 2} [\emph{$(u,v)$ is improving}]:
$w(u,v) > w(u,m(u)) + w(v,m(v))$.

The algorithm has failed to discover this improving edge in $K$ rounds.
Let $l_{uv} = w(u,v) - w(u,m(u)) - w(v,m(v)) > 0$ and note $p_{uv} \geq p_{\min}$.
By Lemma~2 with $K = \Theta(\ln(1/\varepsilon)/p_{\min})$:
\[
    \mathbb{E}[l_{uv}]
    \;\leq\; l_{uv}\,e^{-Kp_{\min}}
    \;\leq\; \varepsilon\,w(u,v).
\]
Since $\mathbb{E}[w(u,m(u)) + w(v,m(v))] = w(u,v) - \mathbb{E}[l_{uv}]$, rearranging gives:
\[
    \mathbb{E}[w(u,m(u)) + w(v,m(v))] \;\geq\; (1-\varepsilon)\,w(u,v).
\]

Both cases establish \eqref{eq:per-edge}. The analysis involves only weights incident to $u$ or $v$; no weights from auxiliary matched vertices are charged to $w(u,m(u)) + w(v,m(v))$.
\end{proof}

\begin{corollary}
If $\beta = \Theta(\ln\Delta \cdot K)$ and $p_{\min}$ is treated as a constant, then:
\[
\varepsilon(K) \;=\; O\!\left(\frac{1}{K}\right) + O(e^{-K}).
\]
\end{corollary}

\begin{proof}
From Lemma~1 with $\beta = \Theta(\ln\Delta \cdot K)$:
\[
\mathbb{E}[\mathbf{SoftmaxError}(u)] \;\leq\; \frac{\ln\Delta}{\beta} \;=\; O\!\left(\frac{1}{K}\right).
\]
From Lemma~2 with constant $p_{\min}$ and $K$ rounds:
\[
\mathbb{E}[\mathbf{StoppingError}(u,v^*)] \;\leq\; w(u,v^*)\,e^{-Kp_{\min}} \;=\; O(e^{-K}).
\]
Combining both terms gives $\varepsilon(K) = O(1/K) + O(e^{-K})$.
\end{proof}

\subsection{Runtime Analysis}

We assume that the alias method of Vose~\cite{vose1991alias} is used for sampling, so that each alias table is built in $O(\deg(u))$ time and each proposal is drawn in $O(1)$. Preprocessing all tables therefore takes $O(m)$ time in total.

\begin{theorem}[Runtime]
    The Random Proposals Algorithm with $K$ rounds runs in $O(m + Kn)$ time.
    For the $(1/2-\varepsilon)$ approximation guarantee with $K = \Theta(\ln(1/\varepsilon)/p_{\min})$ rounds, the total runtime is:
    \[
        O\!\left(m + \frac{n\,\ln(1/\varepsilon)}{p_{\min}}\right).
    \]
\end{theorem}

\begin{proof}
    Preprocessing builds all alias tables in $O(m)$ time. In each of the $K$ rounds, every vertex is visited once: sampling a neighbor costs $O(1)$ via the alias table, and the weight comparison is $O(1)$. The per-round cost is therefore $O(n)$, giving a total of $O(m) + K \cdot O(n) = O(m + Kn)$. Substituting $K = \Theta(\ln(1/\varepsilon)/p_{\min})$ yields the stated bound.
\end{proof}

\paragraph{Discussion of $p_{\min}$.}
For the softmax distribution with parameter $\beta$ and edge weights in $[1, W]$, the minimum proposal probability satisfies:
\[
    p_{\min} \;=\; \min_{(u,v)\in E} P(u \to v)
    \;\geq\; \frac{e^{\beta}}{\Delta\,e^{\beta W}}
    \;=\; \frac{e^{-\beta(W-1)}}{\Delta}.
\]
This lower bound can be exponentially small in $\beta$ and $W$, reflecting the fact that a highly peaked distribution may rarely propose low-weight edges. Consequently, the number of rounds required in the worst case grows as $e^{\beta(W-1)}$.

\paragraph{Constant-$p_{\min}$ assumption.}
In practice, when $\beta$ and the weight range $W$ are treated as fixed constants independent of $n$, $p_{\min} = \Omega(1)$ and the required $K$ reduces to $O(\ln(1/\varepsilon))$. Under this mild assumption the runtime simplifies to $O(m + n\log(1/\varepsilon))$, which equals $O(m\log(1/\varepsilon))$ since $n \leq 2m$ in any connected component. This matches the runtime of the best known randomized $2/3$-approximation algorithms.

\section{Experimental Evaluation}

We evaluated the proposed Random Proposals algorithm against Greedy and Suitor, two established $1/2$-approximation baselines for MWM. All experiments were conducted on the KFUPM High Performance Computing Center using CPU resources. Each job ran on a single compute node with one task, 4 CPU cores, and 32~GB of memory.

The evaluation aims to answer three questions: (1) how does RP's matching quality compare to deterministic baselines across different graph families; (2) how do the parameters $K$ and $\beta$ influence quality and runtime; and (3) what is the practical quality-runtime tradeoff offered by RP. We tested all combinations of $K \in \{1, 2, 3, 6, 10, 15, 18, 20\}$ and $\beta \in \{12, 25, 50, 75, 100, 115\}$.

The benchmark graphs were drawn from the SuiteSparse Matrix Collection with randomly generated weights $w_i \in [1,1000]$. We also generated synthetic graphs using the Barab\'{a}si--Albert model~\cite{barabasi1999emergence} (ab305, ab205, ab105) and the Erd\H{o}s--R\'{e}nyi model~\cite{erdos1959random} (er1). These families represent a wide range of structural properties, from sparse road networks to scale-free social graphs to high-dimensional finite-element meshes. For all benchmark graphs, we computed exact maximum-weight matchings using Edmonds'~\cite{edmonds1965paths} and Gabow's~\cite{gabow1990data} algorithms to serve as an absolute quality reference.

\paragraph{Overall quality.}
Table~\ref{tab:summary} summarizes the best RP configuration (always at $K=20$) against both baselines. RP is competitive in matching quality across all graphs and outperforms Greedy and Suitor on nine of the thirteen instances. The improvements are particularly pronounced on graphs with irregular, scale-free degree distributions. On \texttt{kron\_g500-logn21}, RP reaches $86.96\%$ of the exact optimum, compared with $78.14\%$ for Greedy and $78.18\%$ for Suitor — a gap of nearly nine percentage points. On \texttt{com-LiveJournal}, RP achieves $90.30\%$ versus $85.12\%$ and $85.07\%$ for the baselines. The Barab\'{a}si--Albert instances (ab105, ab205, ab305) also show consistent improvements of roughly $3\%$ over both deterministic methods.

\paragraph{Effect of graph structure.}
The variation in RP's relative advantage across graph families reveals an important structural pattern. On scale-free graphs with heavy-tailed degree distributions — such as \texttt{kron\_g500-logn21}, \texttt{com-LiveJournal}, and the Barab\'{a}si--Albert instances — deterministic algorithms tend to perform poorly because high-degree hub vertices attract many competing proposals. Greedy and Suitor resolve such conflicts through fixed local rules that systematically exclude certain heavy edges, and the resulting matching may be far from optimal. RP's stochastic sampling, by contrast, occasionally discovers alternative configurations in which hub edges are matched differently, enabling recovery of weight that deterministic rules would permanently forgo. The fact that all three Barab\'{a}si--Albert instances (ab105, ab205, ab305) yield nearly identical improvement percentages ($91.21\%$ vs.\ $88.35\%$) confirms that the gain is a structural effect of the power-law degree distribution rather than a consequence of graph size.

On sparse road-network graphs (\texttt{europe\_osm}, \texttt{asia\_osm}), the situation is different. These graphs have low, bounded degree and a geometric structure in which most edges have locally unique heavy partners. Greedy and Suitor already capture near-optimal matches in this setting, and RP offers only a small additional gain (at most $0.4\%$). The random sampling adds little benefit when the local structure is simple enough that deterministic rules already succeed.

There are also graphs where RP underperforms the baselines. On \texttt{nlpkkt200} (a finite-element mesh) and \texttt{rgg\_n\_2\_23\_s0} (a random geometric graph), Greedy and Suitor exceed RP even at $K=20$. Both graphs have regular local structure and geographically clustered weights, so the greedy matching already captures most of the heavy edges without needing randomization. In these cases, RP's stochastic exploration explores configurations of similar quality without systematically improving upon the deterministic solution.

\paragraph{Effect of $K$.}
Tables~\ref{tab:kron}--\ref{tab:rgg} and Figure~\ref{fig:weight_vs_K_exact} show the full $(K,\beta)$ parameter sweep for four representative graphs. The dominant trend is that increasing $K$ consistently improves matching quality, with the rate of improvement diminishing as $K$ grows. On \texttt{kron\_g500-logn21}, quality rises from $47.46\%$ at $K=1$ to $75.65\%$ at $K=6$ and then more slowly to $86.96\%$ at $K=20$. On \texttt{com-LiveJournal}, quality passes the Greedy/Suitor level ($\approx 85\%$) at around $K=6$ and continues to improve moderately thereafter. This pattern is consistent with the theoretical bound from the Corollary: the $O(e^{-K})$ term governs the rapid early improvement, while the $O(1/K)$ term governs the slower saturation at large $K$. In practical terms, $K=6$ already achieves more than $96\%$ of the exact optimum on sparse graphs, and $K=10$ is generally sufficient to outperform the deterministic baselines on scale-free graphs.

\paragraph{Effect of $\beta$.}
In contrast to $K$, the bias parameter $\beta$ has a much weaker effect on quality. For any fixed $K$, the variation in matching weight across the six tested $\beta$ values is at most a few tenths of a percent on all graphs. This is consistent with the theoretical analysis: the softmax error $\mathbb{E}[\mathbf{SoftmaxError}(u)] \leq \ln\Delta/\beta$ decreases only inversely in $\beta$, while the stopping error decreases exponentially in $K$. At the tested values, even the smallest $\beta = 12$ provides sufficient concentration toward heavy edges; further increases in $\beta$ yield negligible gains. Accordingly, $\beta$ is best treated as a secondary fine-tuning parameter once $K$ is set.

\paragraph{Runtime tradeoff.}
Figure~\ref{fig:time_vs_K} shows that total runtime grows approximately linearly in $K$, confirming the theoretical $O(m + Kn)$ bound. RP is generally slower than Suitor and may be faster or slower than Greedy depending on the graph and the chosen $K$. On \texttt{kron\_g500-logn21}, RP at $K=20$ is substantially faster than Greedy ($9.53$\,s vs.\ $37.91$\,s) despite producing significantly better quality, though it is slower than Suitor ($2.88$\,s). On \texttt{com-LiveJournal}, RP at $K \in \{1,2\}$ actually runs faster than Suitor while already matching or exceeding Greedy quality (see Table~\ref{tab:livejournal}), offering a compelling operating point for time-sensitive applications.

Because matching quality saturates logarithmically in $K$ while runtime grows linearly, the marginal cost per unit of quality improvement increases with $K$. This makes moderate values of $K$ — roughly $K \in \{6, 10\}$ — the most cost-effective operating regime across the majority of graphs tested.

\begin{table}[ht]
\centering
\caption{%
  Comparison of RP (best configuration, always at $K{=}20$),
  Greedy, and Suitor across all benchmark graphs.
  All weights are reported as a percentage of the exact maximum-weight matching.
  RP time is the matching time only (excluding preprocessing).
  Best~$\beta$ is the value that achieved the highest RP weight.
}
\label{tab:summary}
\setlength{\tabcolsep}{4pt}
\resizebox{\textwidth}{!}{%
\begin{tabular}{l r S[table-format=2.2] S[table-format=2.2] S[table-format=2.2] r S[table-format=3.2] S[table-format=4.3] S[table-format=4.4]}
\toprule
 & & \multicolumn{3}{c}{\textbf{Weight (\% of exact optimum)}}
 & & \multicolumn{3}{c}{\textbf{Time (s)}} \\
\cmidrule(lr){3-5}\cmidrule(lr){7-9}
\textbf{Dataset}
  & \textbf{Exact weight}
  & {\textbf{RP}}
  & {\textbf{Greedy}}
  & {\textbf{Suitor}}
  & {\textbf{Best $\beta$}}
  & {\textbf{RP}}
  & {\textbf{Greedy}}
  & {\textbf{Suitor}} \\
\midrule
\multicolumn{9}{l}{} \\[2pt]
kron\_g500-logn21   & $3.513\times10^{8}$  & 86.96 & 78.14 & 78.18 & 12  &   9.53 &  37.907 &  2.8819 \\
europe\_osm         & $1.495\times10^{10}$ & 96.72 & 96.29 & 96.29 & 115 & 113.78 &  24.050 &  4.0690 \\
com-LiveJournal     & $1.260\times10^{9}$  & 90.30 & 85.12 & 85.07 & 12  &  18.29 &  13.596 &  2.4346 \\
rgg\_n\_2\_23\_s0    & $3.746\times10^{9}$  & 89.46 & 92.67 & 92.68 & 12  &  23.26 &  27.675 &  3.2798 \\
mtx2010             & $3.097\times10^{8}$  & 93.28 & 91.66 & 91.66 & 25  &   2.60 &   0.803 &  0.1414 \\
kmer\_V2a           & $1.620\times10^{10}$ & 96.87 & 96.34 & 96.34 & 25  & 161.87 &  27.632 &  9.3311 \\
asia\_osm           & $3.518\times10^{9}$  & 96.72 & 96.32 & 96.32 & 115 &  28.25 &   5.280 &  0.8119 \\
nlpkkt200           & $7.550\times10^{9}$  & 88.38 & 94.31 & 94.31 & 12  &  69.99 & 101.727 & 10.6408 \\
hugebubbles-00010   & $6.445\times10^{9}$  & 95.05 & 94.90 & 94.90 & 25  &  69.95 &  12.999 &  4.6539 \\
ab305               & $1.135\times10^{10}$ & 91.21 & 88.35 & 88.38 & 12  & 157.46 &  71.973 & 22.1548 \\
ab205               & $7.571\times10^{9}$  & 91.21 & 88.36 & 88.39 & 12  &  96.11 &  53.455 & 15.3804 \\
ab105               & $3.786\times10^{9}$  & 91.21 & 88.35 & 88.38 & 12  &  45.90 &  26.438 &  6.7433 \\
er1                 & $4.452\times10^{8}$  & 82.12 & 77.16 & 77.84 & 12  &   5.55 &  20.733 &  1.9945 \\
\bottomrule
\end{tabular}%
}
\end{table}


\bigskip\noindent
Tables~\ref{tab:kron}--\ref{tab:rgg} show the full $(K,\beta)$ parameter
sweep for four representative graphs with known exact optima.
Each cell reports the RP matching weight as a percentage of the exact optimum,
with the matching time (in seconds) shown in parentheses.
Greedy and Suitor are listed as reference baselines at the bottom of each table.

\begin{table}[ht]
\centering
\caption{%
  \texttt{kron\_g500-logn21}: RP weight as \% of exact optimum
  ($3.513\times10^{8}$) for all $(K,\beta)$ combinations.
  Matching time (s) in parentheses.
  Greedy: $78.14\%$ in $37.91$\,s.
  Suitor: $78.18\%$ in $2.88$\,s.
}
\label{tab:kron}
\setlength{\tabcolsep}{5pt}
\resizebox{\linewidth}{!}{%
\begin{tabular}{r llllll}
\toprule
\multirow{2}{*}{$K$} & \multicolumn{6}{c}{$\beta$} \\
\cmidrule(lr){2-7}
 & \multicolumn{1}{c}{12} & \multicolumn{1}{c}{25} & \multicolumn{1}{c}{50}
 & \multicolumn{1}{c}{75} & \multicolumn{1}{c}{100} & \multicolumn{1}{c}{115} \\
\midrule
 1 & 47.46\% (0.49s) & 46.91\% (0.49s) & 46.65\% (0.48s) & 46.55\% (0.48s) & 46.51\% (0.47s) & 46.49\% (0.48s) \\
 2 & 58.90\% (0.92s) & 58.40\% (0.93s) & 58.16\% (0.93s) & 58.07\% (0.92s) & 58.04\% (0.93s) & 58.02\% (0.91s) \\
 3 & 65.54\% (1.40s) & 65.10\% (1.38s) & 64.90\% (1.37s) & 64.81\% (1.36s) & 64.78\% (1.36s) & 64.77\% (1.36s) \\
 6 & 75.65\% (2.84s) & 75.35\% (2.53s) & 75.20\% (2.75s) & 75.14\% (2.52s) & 75.12\% (2.79s) & 75.11\% (2.53s) \\
10 & 81.44\% (4.79s) & 81.25\% (5.73s) & 81.15\% (4.59s) & 81.11\% (5.49s) & 81.09\% (4.54s) & 81.09\% (4.61s) \\
15 & 85.00\% (7.39s) & 84.89\% (6.06s) & 84.83\% (7.00s) & 84.80\% (6.10s) & 84.80\% (6.99s) & 84.79\% (6.11s) \\
18 & 86.30\% (8.63s) & 86.21\% (7.45s) & 86.17\% (8.49s) & 86.14\% (7.52s) & 86.13\% (8.69s) & 86.13\% (7.36s) \\
20 & 86.96\% (9.53s) & 86.89\% (8.32s) & 86.85\% (8.13s) & 86.83\% (8.43s) & 86.82\% (8.11s) & 86.81\% (8.65s) \\
\midrule
\multicolumn{7}{l}{\textit{Reference baselines}} \\
Greedy & \multicolumn{6}{l}{78.14\% \quad (37.91\,s)} \\
Suitor & \multicolumn{6}{l}{78.18\% \quad (2.88\,s)} \\
\bottomrule
\end{tabular}
}
\end{table}

\begin{table}[ht]
\centering
\caption{%
  \texttt{europe\_osm}: RP weight as \% of exact optimum
  ($1.495\times10^{10}$) for all $(K,\beta)$ combinations.
  Matching time (s) in parentheses.
  Greedy: $96.29\%$ in $24.05$\,s.
  Suitor: $96.29\%$ in $4.07$\,s.
}
\label{tab:europe}
\setlength{\tabcolsep}{5pt}
\resizebox{\linewidth}{!}{%
\begin{tabular}{r llllll}
\toprule
\multirow{2}{*}{$K$} & \multicolumn{6}{c}{$\beta$} \\
\cmidrule(lr){2-7}
 & \multicolumn{1}{c}{12} & \multicolumn{1}{c}{25} & \multicolumn{1}{c}{50}
 & \multicolumn{1}{c}{75} & \multicolumn{1}{c}{100} & \multicolumn{1}{c}{115} \\
\midrule
 1 & 84.36\% (7.45s)   & 83.27\% (7.40s)   & 82.76\% (7.31s)   & 82.57\% (7.28s)   & 82.49\% (6.87s)   & 82.45\% (6.88s)   \\
 2 & 92.84\% (14.08s)  & 92.46\% (13.81s)  & 92.28\% (13.70s)  & 92.21\% (12.70s)  & 92.18\% (12.61s)  & 92.17\% (12.65s)  \\
 3 & 95.34\% (24.28s)  & 95.23\% (20.24s)  & 95.17\% (22.84s)  & 95.14\% (22.59s)  & 95.13\% (19.96s)  & 95.13\% (19.99s)  \\
 6 & 96.60\% (41.75s)  & 96.63\% (45.76s)  & 96.63\% (38.77s)  & 96.63\% (38.64s)  & 96.63\% (35.04s)  & 96.63\% (35.10s)  \\
10 & 96.67\% (65.40s)  & 96.71\% (58.06s)  & 96.71\% (69.20s)  & 96.71\% (63.36s)  & 96.71\% (71.42s)  & 96.71\% (63.24s)  \\
15 & 96.68\% (110.37s) & 96.71\% (86.36s)  & 96.72\% (96.01s)  & 96.72\% (94.23s)  & 96.72\% (97.42s)  & 96.72\% (92.04s)  \\
18 & 96.68\% (124.13s) & 96.71\% (114.68s) & 96.72\% (101.76s) & 96.72\% (105.13s) & 96.72\% (126.34s) & 96.72\% (100.83s) \\
20 & 96.68\% (129.71s) & 96.71\% (126.75s) & 96.72\% (112.93s) & 96.72\% (113.00s) & 96.72\% (112.35s) & 96.72\% (113.78s) \\
\midrule
\multicolumn{7}{l}{\textit{Reference baselines}} \\
Greedy & \multicolumn{6}{l}{96.29\% \quad (24.05\,s)} \\
Suitor & \multicolumn{6}{l}{96.29\% \quad (4.07\,s)} \\
\bottomrule
\end{tabular}
}
\end{table}

\begin{table}[ht]
\centering
\caption{%
  \texttt{com-LiveJournal}: RP weight as \% of exact optimum
  ($1.260\times10^{9}$) for all $(K,\beta)$ combinations.
  Matching time (s) in parentheses.
  Greedy: $85.12\%$ in $13.60$\,s.
  Suitor: $85.07\%$ in $2.43$\,s.
}
\label{tab:livejournal}
\setlength{\tabcolsep}{5pt}
\resizebox{\linewidth}{!}{%
\begin{tabular}{r llllll}
\toprule
\multirow{2}{*}{$K$} & \multicolumn{6}{c}{$\beta$} \\
\cmidrule(lr){2-7}
 & \multicolumn{1}{c}{12} & \multicolumn{1}{c}{25} & \multicolumn{1}{c}{50}
 & \multicolumn{1}{c}{75} & \multicolumn{1}{c}{100} & \multicolumn{1}{c}{115} \\
\midrule
 1 & 65.92\% (1.04s)  & 65.03\% (1.03s)  & 64.62\% (0.97s)  & 64.47\% (1.03s)  & 64.41\% (1.01s)  & 64.38\% (1.59s)  \\
 2 & 76.06\% (1.94s)  & 75.40\% (1.99s)  & 75.08\% (1.98s)  & 74.97\% (1.97s)  & 74.92\% (1.98s)  & 74.90\% (1.92s)  \\
 3 & 80.67\% (2.79s)  & 80.16\% (2.99s)  & 79.92\% (2.95s)  & 79.83\% (2.97s)  & 79.79\% (2.94s)  & 79.77\% (2.70s)  \\
 6 & 86.17\% (5.80s)  & 85.90\% (5.29s)  & 85.77\% (5.20s)  & 85.72\% (5.69s)  & 85.70\% (5.51s)  & 85.69\% (5.66s)  \\
10 & 88.60\% (9.55s)  & 88.44\% (8.61s)  & 88.37\% (8.52s)  & 88.34\% (9.35s)  & 88.33\% (9.28s)  & 88.32\% (9.35s)  \\
15 & 89.77\% (14.15s) & 89.67\% (12.67s) & 89.62\% (17.30s) & 89.61\% (12.41s) & 89.60\% (14.00s) & 89.59\% (13.94s) \\
18 & 90.13\% (17.11s) & 90.05\% (16.61s) & 90.01\% (15.36s) & 89.99\% (15.05s) & 89.99\% (16.75s) & 89.98\% (16.70s) \\
20 & 90.30\% (18.29s) & 90.23\% (18.77s) & 90.19\% (17.96s) & 90.18\% (16.74s) & 90.17\% (18.57s) & 90.17\% (17.80s) \\
\midrule
\multicolumn{7}{l}{\textit{Reference baselines}} \\
Greedy & \multicolumn{6}{l}{85.12\% \quad (13.60\,s)} \\
Suitor & \multicolumn{6}{l}{85.07\% \quad (2.43\,s)} \\
\bottomrule
\end{tabular}
}
\end{table}

\begin{table}[ht]
\centering
\caption{%
  \texttt{rgg\_n\_2\_23\_s0}: RP weight as \% of exact optimum
  ($3.746\times10^{9}$) for all $(K,\beta)$ combinations.
  Matching time (s) in parentheses.
  Greedy: $92.67\%$ in $27.67$\,s.
  Suitor: $92.68\%$ in $3.28$\,s.
}
\label{tab:rgg}
\setlength{\tabcolsep}{5pt}
\resizebox{\linewidth}{!}{%
\begin{tabular}{r llllll}
\toprule
\multirow{2}{*}{$K$} & \multicolumn{6}{c}{$\beta$} \\
\cmidrule(lr){2-7}
 & \multicolumn{1}{c}{12} & \multicolumn{1}{c}{25} & \multicolumn{1}{c}{50}
 & \multicolumn{1}{c}{75} & \multicolumn{1}{c}{100} & \multicolumn{1}{c}{115} \\
\midrule
 1 & 55.73\% (1.30s)  & 54.95\% (1.28s)  & 54.59\% (1.32s)  & 54.46\% (1.28s)  & 54.40\% (1.39s)  & 54.38\% (1.25s)  \\
 2 & 67.94\% (2.42s)  & 67.34\% (2.43s)  & 67.05\% (2.33s)  & 66.95\% (2.38s)  & 66.91\% (2.42s)  & 66.89\% (2.37s)  \\
 3 & 73.87\% (3.77s)  & 73.40\% (3.23s)  & 73.17\% (3.58s)  & 73.09\% (3.50s)  & 73.06\% (3.60s)  & 73.04\% (3.49s)  \\
 6 & 81.74\% (7.17s)  & 81.46\% (7.07s)  & 81.32\% (7.02s)  & 81.27\% (6.16s)  & 81.25\% (6.72s)  & 81.24\% (6.44s)  \\
10 & 85.85\% (13.41s) & 85.67\% (11.33s) & 85.58\% (11.45s) & 85.56\% (10.07s) & 85.54\% (11.08s) & 85.54\% (10.31s) \\
15 & 88.21\% (17.81s) & 88.10\% (16.73s) & 88.04\% (16.55s) & 88.03\% (14.66s) & 88.02\% (14.64s) & 88.01\% (14.79s) \\
18 & 89.04\% (21.36s) & 88.95\% (19.57s) & 88.91\% (19.70s) & 88.90\% (17.81s) & 88.89\% (17.40s) & 88.88\% (17.77s) \\
20 & 89.46\% (23.26s) & 89.38\% (21.15s) & 89.34\% (19.47s) & 89.33\% (19.49s) & 89.33\% (19.30s) & 89.32\% (19.64s) \\
\midrule
\multicolumn{7}{l}{\textit{Reference baselines}} \\
Greedy & \multicolumn{6}{l}{92.67\% \quad (27.67\,s)} \\
Suitor & \multicolumn{6}{l}{92.68\% \quad (3.28\,s)} \\
\bottomrule
\end{tabular}
}
\end{table}

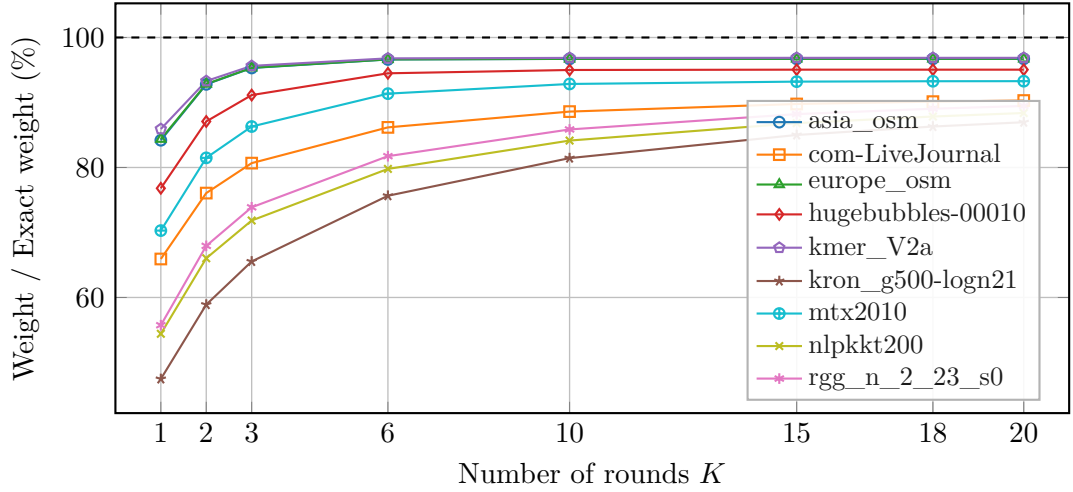
\begin{figure}[!htbp]
\centering
\begin{adjustbox}{max width=\linewidth,max totalheight=0.55\textheight,center}
\begin{tikzpicture}
\begin{axis}[
  width=0.86\textwidth,
  height=7cm,
  xlabel={Number of rounds $K$},
  ylabel={Weight / Exact weight (\%)},
  xmin=0, xmax=21,
  xtick={1,2,3,6,10,15,18,20},
  grid=both,
  grid style={line width=0.3pt, draw=gray!30},
  major grid style={line width=0.5pt, draw=gray!50},
  legend pos=south east,
  legend style={font=\small, fill=white, fill opacity=0.9, draw=gray!60},
  legend cell align=left,
  thick,
]

\addplot[color=col1, mark=o] coordinates {
  (1,84.19) (2,92.79) (3,95.31) (6,96.62)
  (10,96.71) (15,96.72) (18,96.72) (20,96.72)
};
\addlegendentry{asia\_osm}

\addplot[color=col2, mark=square] coordinates {
  (1,65.92) (2,76.06) (3,80.67) (6,86.17)
  (10,88.60) (15,89.77) (18,90.13) (20,90.30)
};
\addlegendentry{com-LiveJournal}

\addplot[color=col3, mark=triangle] coordinates {
  (1,84.36) (2,92.84) (3,95.34) (6,96.63)
  (10,96.71) (15,96.72) (18,96.72) (20,96.72)
};
\addlegendentry{europe\_osm}

\addplot[color=col4, mark=diamond] coordinates {
  (1,76.80) (2,87.08) (3,91.14) (6,94.49)
  (10,95.00) (15,95.05) (18,95.05) (20,95.05)
};
\addlegendentry{hugebubbles-00010}

\addplot[color=col5, mark=pentagon] coordinates {
  (1,85.92) (2,93.30) (3,95.63) (6,96.79)
  (10,96.86) (15,96.87) (18,96.87) (20,96.87)
};
\addlegendentry{kmer\_V2a}

\addplot[color=col6, mark=star] coordinates {
  (1,47.46) (2,58.90) (3,65.54) (6,75.65)
  (10,81.44) (15,85.00) (18,86.30) (20,86.96)
};
\addlegendentry{kron\_g500-logn21}

\addplot[color=col7, mark=oplus] coordinates {
  (1,70.29) (2,81.48) (3,86.30) (6,91.36)
  (10,92.85) (15,93.21) (18,93.27) (20,93.28)
};
\addlegendentry{mtx2010}

\addplot[color=col8, mark=x] coordinates {
  (1,54.41) (2,66.03) (3,71.84) (6,79.78)
  (10,84.14) (15,86.82) (18,87.84) (20,88.38)
};
\addlegendentry{nlpkkt200}

\addplot[color=col9, mark=asterisk] coordinates {
  (1,55.73) (2,67.94) (3,73.87) (6,81.74)
  (10,85.85) (15,88.21) (18,89.04) (20,89.46)
};
\addlegendentry{rgg\_n\_2\_23\_s0}

\addplot[dashed, black] coordinates {(0,100)(21,100)};

\end{axis}
\end{tikzpicture}
\end{adjustbox}
\caption{%
  Normalized matching weight (RP / exact $\times\,100\%$) as a function of $K$, at the best $\beta$ for each graph.
  Large sparse graphs (\texttt{asia\_osm}, \texttt{europe\_osm}, \texttt{kmer\_V2a})
  approach $97\%$ of the exact optimum by $K=6$, while scale-free graphs
  (\texttt{kron\_g500-logn21}) and structured meshes (\texttt{nlpkkt200}) converge more slowly.
  The concave shape is consistent with the theoretical bound $\varepsilon(K) = O(1/K) + O(e^{-K})$.
  The dashed line marks the exact optimum.
}
\label{fig:weight_vs_K_exact}
\end{figure}

\begin{figure}[!htbp]
\centering
\begin{adjustbox}{max width=\linewidth,max totalheight=0.55\textheight,center}
\begin{tikzpicture}
\begin{axis}[
  width=0.86\textwidth,
  height=7cm,
  xlabel={Number of rounds $K$},
  ylabel={Total time (s)},
  xmin=0, xmax=21,
  xtick={1,2,3,6,10,15,18,20},
  grid=both,
  grid style={line width=0.3pt, draw=gray!30},
  major grid style={line width=0.5pt, draw=gray!50},
  legend pos=north west,
  legend style={font=\small, fill=white, fill opacity=0.9, draw=gray!60},
  legend cell align=left,
  thick,
]

\addplot[color=col1, mark=o, mark size=2.5pt] coordinates {
  (1, 15.740) (2, 14.498) (3, 14.615) (6, 16.614)
  (10, 19.233) (15, 20.525) (18, 21.972) (20, 22.915)
};
\addlegendentry{kron\_g500-logn21}

\addplot[color=col2, mark=square, mark size=2.5pt] coordinates {
  (1, 7.353) (2, 8.054) (3, 9.040) (6, 11.563)
  (10, 15.229) (15, 20.394) (18, 22.287) (20, 23.994)
};
\addlegendentry{com-LiveJournal}

\addplot[color=col3, mark=triangle, mark size=2.5pt] coordinates {
  (1, 39.007) (2, 54.107) (3, 52.913) (6, 74.948)
  (10, 106.451) (15, 150.726) (20, 178.445)
};
\addlegendentry{ab305}

\addplot[color=col4, mark=diamond, mark size=2.5pt] coordinates {
  (1, 43.618) (2, 45.441) (3, 52.259) (6, 61.256)
  (10, 75.151) (15, 100.689) (18, 119.116) (20, 118.511)
};
\addlegendentry{nlpkkt200}

\addplot[color=col5, mark=pentagon, mark size=2.5pt] coordinates {
  (1, 23.336) (2, 29.123) (3, 38.511) (6, 54.799)
  (10, 81.176) (15, 111.760) (18, 128.209) (20, 133.420)
};
\addlegendentry{europe\_osm}

\addplot[color=col6, mark=star, mark size=2.5pt] coordinates {
  (1, 13.494) (2, 14.747) (3, 15.539) (6, 18.910)
  (10, 23.747) (15, 28.054) (18, 31.147) (20, 32.237)
};
\addlegendentry{rgg\_n\_2\_23\_s0}

\end{axis}
\end{tikzpicture}
\end{adjustbox}
\caption{%
  Total RP runtime (preprocessing + matching) as a function of $K$,
  averaged over all $\beta$ values. Runtime grows approximately linearly with $K$,
  consistent with the $O(m + Kn)$ bound.
  Larger or denser graphs (\texttt{ab305}, \texttt{europe\_osm}) exhibit steeper scaling.
  Because quality saturates logarithmically while runtime grows linearly, intermediate
  values of $K$ offer the best cost-effectiveness.
}
\label{fig:time_vs_K}
\end{figure}
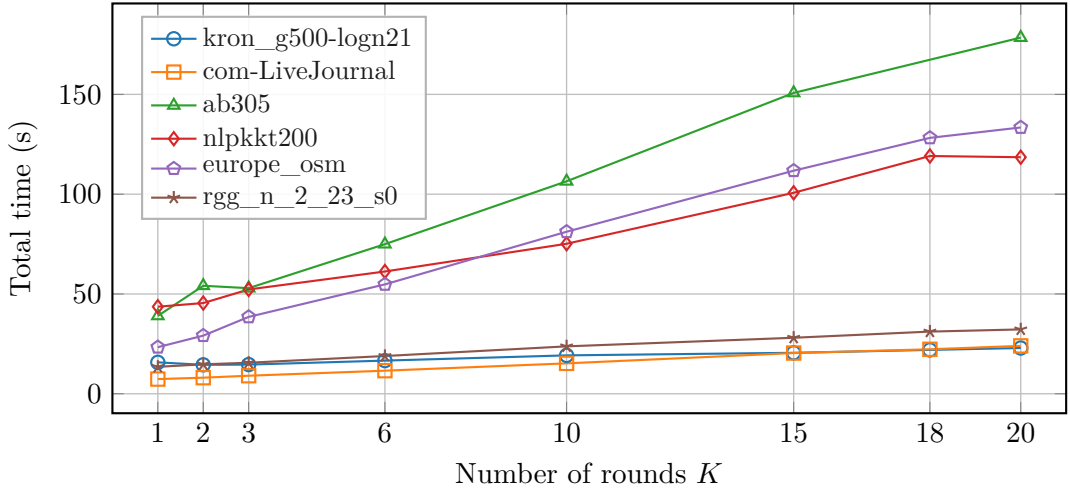
\FloatBarrier

\section{Conclusion}

We presented Random Proposals (RP), a randomized softmax-based local-improvement algorithm for the maximum weighted matching problem. The algorithm introduces a tunable mechanism that balances exploration and exploitation through the parameters $K$ (number of rounds) and $\beta$ (softmax temperature), providing explicit control over the quality-runtime tradeoff.

Theoretically, we showed that RP achieves an expected $(1/2-\varepsilon)$ approximation ratio with $K = \Theta(\ln(1/\varepsilon)/p_{\min})$ rounds and $\beta = O(\ln\Delta/\varepsilon)$, and that its total runtime is $O(m + Kn)$, which simplifies to $O(m\log(1/\varepsilon))$ under the mild constant-$p_{\min}$ assumption. The approximation proof reduces to a clean two-case analysis: optimal edges that are already non-improving are handled deterministically, while improving edges that remain undiscovered are bounded using the geometric stopping-error lemma.

Experimentally, RP consistently achieves competitive performance and outperforms both Greedy and Suitor on the majority of the benchmarks tested, with the largest gains on graphs with irregular, scale-free degree distributions where deterministic local rules are most constrained. On regular or geometrically structured graphs, deterministic baselines may retain better quality at lower cost. The parameter $K$ is the primary quality control knob: increasing $K$ yields logarithmic quality gains at linear runtime cost, making moderate values of $K$ most cost-effective in practice.

Future work includes developing adaptive strategies for online selection of $K$ and $\beta$, improving robustness across diverse graph families, and exploring parallel or distributed implementations to further enhance scalability.


\clearpage
\bibliographystyle{unsrt}
\bibliography{references}

\end{document}